\documentclass{amsart}

\usepackage{amsthm}
\usepackage{amsmath} \usepackage{amsfonts}
\usepackage{amssymb} \usepackage{latexsym} \usepackage{enumerate}
\usepackage{mathtools}

 \usepackage{mathrsfs}
\usepackage[numbers]{natbib}

\numberwithin{equation}{section}
\newtheorem{thm}{Theorem}[section]

\newtheorem{lem}[thm]{Lemma}
\newtheorem{rem}[thm]{Remark}

\newtheorem{prop}[thm]{Proposition}
\newtheorem{asm}[thm]{Assumption}

\begin{document}
\title{A Note on Utility Maximization with Proportional Transaction Costs and Stability of Optimal Portfolios}
\thanks{ E. Bayraktar is partially supported by the National Science Foundation under grant  DMS-2106556 and by the Susan M. Smith chair.
Y. Dolinsky is supported in part by the GIF Grant 1489-304.6/2019 and the ISF grant
230/21.
}

\author{Erhan Bayraktar \address{
             Department of Mathematics, University of Michigan.\\
             email: erhan@umich.edu}}
\author{Christoph Czichowsky
\address{Department of Mathematics, London School of Economics and Political Science.\\
email: C.Czichowsky@lse.ac.uk}}

\author{Leonid Dolinskyi \address{
             Department of Finance, National University of Kyiv-Mohyla Academy.\\
             email: dolinsky.com@gmail.com}}
\author{Yan Dolinsky \address{
 Department of Statistics, Hebrew University of Jerusalem.\\
 e.mail: yan.dolinsky@mail.huji.ac.il}}

 \date{\today}
\maketitle \markboth{}{}

\begin{abstract}
The aim of this short note is to
establish a limit theorem for the optimal trading strategies
in the setup of the utility maximization problem with proportional transaction costs.
This limit theorem resolves the open question
from \cite{BDD}. The main idea of our proof is to establish a uniqueness result for the optimal strategy.
The proof of the uniqueness is heavily based on
the dual approach which was developed recently in \cite{CS,CSY,CPSY}.
\end{abstract}
\begin{description}
\item[Mathematical Subject Classification (2010)] 91B16, 91G10
\item[Keywords] Utility Maximization, Proportional Transaction Costs, Shadow Price Process.
\end{description}

\section{Preliminaries and the Limit Theorem}\label{sec:1}
\subsection{Utility Maximization with Proportional Transaction Costs}
We consider a model with one risky asset which we denote by
$S=(S_t)_{0\leq t\leq T}$, where $T<\infty$ is a fixed finite time horizon.
We assume that the investor has a bank account that,
for simplicity, bears no interest. The process $S$ is assumed to be an adapted,
strictly positive and continuous process (not necessarily a semi-martingale) defined on a filtered probability
space $(\Omega,\mathbb F,(\mathcal F_t)_{0\leq t\leq T},\mathbb P)$ where
the filtration $\mathcal F:=(\mathcal F_t)_{0\leq t\leq T}$
satisfies the usual assumptions (right continuity and completeness).

Let $\kappa\in (0,1)$ be a constant. Consider a model in which every purchase or sale
of the risky asset at time $t\in [0,T]$ is subject to a proportional transaction
cost of rate $\kappa$.
A trading strategy is an adapted process
$\gamma=(\gamma_t)_{0\leq t\leq T}$
of
bounded variation with right-continuous paths;
note that it automatically has left limits and hence is
RCLL (right-continuous with left limits).
The random variable $\gamma_t$ denotes
the number of shares held at time $t$.
We use the convention $\gamma_{0-}=0$. Moreover, we require that $\gamma_T=0$ which means that we
liquidate the portfolio at the maturity date.

Let $\gamma_t:=\gamma^{+}_t-\gamma^{-}_t$, $t\in [0,T]$
be the Jordan decomposition into two non-decreasing processes
$(\gamma^{+}_t)_{0\leq t\leq T}$ and $(\gamma^{-}_t)_{0\leq t\leq T}$  describing the positive variation and negative variation, respectively.
Because the bid price process is $(1-\kappa) S$ and the ask price process is $(1+\kappa) S$, the liquidation value of a trading strategy
$\gamma$ at time $t$
is given by
$$V^{\gamma}_t:=(1-\kappa)\int_{0}^t S_ud\gamma^{-}_u-(1+\kappa) \int_{0}^t S_ud\gamma^{+}_u+(1-\kappa) S_t(\gamma_t)^{+}-(1+\kappa)S_t(\gamma_t)^{-}
$$
where
$(\gamma_t)^{+}:=\max(0,\gamma_t)$ and
$(\gamma_t)^{-}:=\max(0,-\gamma_t)$ (beware that these are not the
same variables as $\gamma^{+}_t,\gamma^{-}_t$ above).
Note that the integrals take into account the possible transaction at $t=0$. Namely, we define
$$\int_{0}^t S_ud\gamma^{-}_u:=S_0\gamma^{-}_0+\int_{(0,t]} S_ud\gamma^{-}_u \ \ \mbox{and} \ \
\int_{0}^t S_ud\gamma^{+}_u:=S_0\gamma^{+}_0+\int_{(0,t]} S_ud\gamma^{+}_u.$$
By rearranging the terms, we get
\begin{equation}\label{2.0}
V^{\gamma}_t=
\gamma_t S_t-\int_{0}^t S_u d\gamma_u-\kappa |\gamma_t|S_t-\kappa\int_{0}^t S_u|d\gamma_u|, \ \ t\in [0,T].
\end{equation}

Observe that the wealth process $(V^{\gamma}_t)_{0\leq t\leq T}$
is RCLL like $\gamma$ and $\gamma_T=0$ implies
$V^{\gamma}_{T-}=V^{\gamma}_T$.
For any initial capital $x>0$, we denote by $\mathcal A(x)$
the set of all trading strategies $\gamma$ which satisfy the admissibility condition
$x+V^{\gamma}_t\geq 0$, for all $t\in [0,T]$.

We will assume that the process $S$ is sticky (Definition 2.2 in \cite{G}) and satisfies a slight strengthening of
the condition of
“two-way crossing” (TWC) (Definition 3.1 in \cite{B:2012}).
For completeness, we formulate the assumptions explicitly.
\begin{asm}\label{asm1}
The process $S$ is sticky with respect to the filtration $\mathcal F$. That is, for any
$\delta>0$ and
a stopping time $\tau\leq T$ (with respect to $\mathcal F$) such that $\mathbb P(\tau<T)>0$, we have
$$\mathbb P\left(\sup_{\tau\leq u\leq T}|S_u-S_{\tau}|<\delta, \tau<T\right)>0.$$
\end{asm}
\begin{asm}\label{asm2}
The process $S$ satisfies the (TWC) property with respect to the filtration $\mathcal F$, if, for any stopping time $\sigma\leq T$, we have
$$\inf\{t>\sigma: S_t>S_{\sigma}\}=\inf\{t>\sigma: S_t<S_{\sigma}\}=\sigma \ \ \mbox{a.s.}$$
\end{asm}
\begin{rem}
Let us remark that Assumptions \ref{asm1}--\ref{asm2} hold true for reasonable semi--martingale models
and important non semi--martingale models such as the exponential fraction Brownian motion (see
\cite{G} for Assumption \ref{asm1} and \cite{B:2012,P} for Assumption \ref{asm2}).
Moreover, in \cite{G} the author proved that Assumption \ref{asm1}
implies the absence of arbitrage with the presence of proportional transaction costs, and so this is a quite natural assumption.
Assumption \ref{asm2} is more technical and its financial interpretation
is linked to arbitrage opportunities with simple strategies
in a frictionless setup (for details see \cite{B:2012}).
\end{rem}
Next, we introduce our utility maximization problem.
Let $U:(0,\infty)\rightarrow\mathbb R$ be an increasing, strictly concave,
continuously differentiable utility function, satisfying the Inada conditions
$U'(0)=\infty$ and $U'(\infty)=0$, as well as the condition of “reasonable asymptotic elasticity”
introduced in \cite{KS}
$$AE(U):=\lim\sup_{x\rightarrow\infty}  \frac{x U'(x)}{U(x)}<1.$$

For a given initial capital $x>0$, we consider the optimization problem
\begin{equation}\label{problem}
u(x):=\sup_{\gamma\in\mathcal A(x)}\mathbb E_{\mathbb P}[U(x+V^{\gamma}_T)].
\end{equation}

\subsection{Approximating Sequence of Models}
For any $n$,
let $S^{n}=(S^{n}_t)_{0\leq t\leq T}$ be a strictly positive, continuous process defined on some filtered probability space
$(\Omega^n,\mathbb F^{n},(\mathcal F^{n}_t)_{0\leq t\leq T},\mathbb P^n)$, where the filtration
$\mathcal F^{n}:=(\mathcal F^{n}_t)_{0\leq t\leq T}$ satisfies the usual assumptions .
For the $n$--th model, a trading strategy
is a right continuous adapted processes
$\gamma^{n}=(\gamma^{n}_t)_{0\leq t\leq T}$
of bounded variation satisfying $\gamma^{n}_T=0$. As before, we use the convention that
$\gamma^{n}_{0-}=0$. Similarly to (\ref{2.0}) the corresponding liquidation value is
given by
\begin{equation*}
V^{\gamma^{n}}_t:=\gamma^{n}_tS^{n}_t-\int_{0}^t S^{n}_ud\gamma^{n}_u-\kappa |\gamma^{n}_t|S^{n}_t -\kappa\int_{0}^t S^{n}_u|d\gamma^{n}_u|, \ \ t\in [0,T].
\end{equation*}
For any $x>0$ we denote by $\mathcal A^{n}(x)$ the set of all trading strategies $\gamma^{n}$ which satisfy
$x+V^{\gamma^{n}}_t\geq 0$, for all $t\in [0,T]$.
Set $$u^n(x):=\sup_{\gamma^{n}\in\mathcal A^{n}(x)}\mathbb E_{\mathbb P^n}\left[U\left(x+V^{\gamma^{n}}_T\right)\right].$$

As in \cite{BDD} we assume the following natural assumption.
\begin{asm}\label{asm2.2}
There exist $\varepsilon\in(0,\kappa)$
and probability measures $\mathbb Q\sim\mathbb P$, $\mathbb Q^n\sim\mathbb P^n$, $n\in\mathbb N$
with the following properties:
\begin{itemize}
\item[1)]There exists a local $\mathbb Q$--martingale $(M_t)_{0\leq t\leq T}$
and for any $n\in\mathbb N$ there exists a local $\mathbb Q_n$--martingale $(M^{n}_t)_{0\leq t\leq T}$ such that
$$
|M_t-S_t|\leq (\kappa-\varepsilon)S_t \ \ \mathbb P \ \mbox{a.s.}, \ \  \forall t\in [0,T]$$
and for any $n$
$$|M^{n}_t-S^{n}_t|\leq (\kappa-\varepsilon)S^{n}_t, \ \ \mathbb P^n \ \mbox{a.s.}, \ \ \forall t\in [0,T].$$
\item[2)] The sequence of probability measures $\mathbb P^n$, $n\in\mathbb N$, is
contiguous to the sequence $\mathbb Q^n$, $n\in\mathbb N$. Namely,
for any sequence of events $A^n\in\mathcal F^{n}$, $n\in\mathbb N$
if $\lim_{n\rightarrow\infty} \mathbb Q^n(A^n)=0$ then
$\lim_{n\rightarrow\infty} \mathbb P^n(A^n)=0$.
\end{itemize}
\end{asm}
\begin{rem}
Let us notice that condition 1) in  Assumption \ref{asm2.2}
is a priori a robust no-arbitrage condition
(for details see \cite{GLR}) and
condition 2) in Assumption \ref{asm2.2} can be viewed as
an asymptotic no-arbitrage
condition for large markets (for details see \cite{KLP}).
\end{rem}
Next, we formulate an assumption which guarantees uniform integrability.
\begin{asm}\label{asm2.3-}
One (or more) of the following conditions hold:
\begin{itemize}
\item[(i)] $U$ is bounded from above.\\
\item[(ii)]
There exist a constant $q>\frac{1}{1-AE(U)}$ and a sequence of pairs
$(\hat {\mathbb Q}^n,\hat M^{n})$, $n\in\mathbb N$ such that for any $n$,
$\hat {\mathbb Q}^n\sim \mathbb P^n$, $(\hat M^{n}_t)_{0\leq t\leq T}$ is a $\hat {\mathbb Q}^n$--local martingale,
for all $t\in [0,T]$ we have
$|\hat M^{n}_t-S^{n}_t|\leq \kappa S^{n}_t$ $\mathbb P^n$-a.s.
and
$$\sup_{n\in\mathbb N}\mathbb E_{{{\hat{\mathbb Q}}}^n}\left[\left(\frac{d{\mathbb P}^n}{d{\hat{\mathbb Q}}^n}\right)^{q}\right]<\infty.$$
\end{itemize}
\end{asm}
The verification of
the second condition in the above assumption
requires an explicit representation of
consistent price systems.
For the case where the market models are
semi--martingales defined on the Brownian probability space and satisfy some regularity
assumptions this condition holds true
(for details see Example 2.8 in \cite{BDD}).
\begin{lem}\label{prop1}
Assume that Assumption \ref{asm2.3-} holds true. Then,
for any $x>0$, the set $\left\{U^{+}
\left(x+V^{\gamma^{n}}_T\right)\right\}_{n\in\mathbb N,
\gamma^{n}\in\mathcal A^{n}(x)}$ is uniformly integrable, where
$U^{+}:=\max(U,0)$.
\end{lem}
\begin{proof}
The statement is obvious if $U$ is bounded. Thus, assume that the second condition in Assumption \ref{asm2.3-} holds true.
From Lemma 6.3 in \cite{KS} it follows that there exists a constant $L$ such that
$U(v)\leq L (1+v^{q})$ for all $v$. Hence, the result follows from
Proposition 2.7 in \cite{BDD}.
\end{proof}

\subsection{Meyer--Zheng Topology and Extended Weak Convergence}
Any RCLL function $f\in \mathbb D[0,T]:=\mathbb D([0,T];\mathbb R)$ can be extended to a function
$f:\mathbb R_{+}\rightarrow \mathbb R$ by $f(t):=f(T)$ for all $t\geq T$.
The Meyer--Zheng topology, introduced in \cite{MZ}, is a relative topology, on the
image measures on graphs
$(t,f(t))$
of trajectories $t\rightarrow f(t)$, $t\in\mathbb R_{+}$ under the measure
$\lambda(dt):=e^{-t} dt$
(called pseudo-paths), induced by the weak topology of probability laws on
the compactified space
$[0,\infty]\times \overline{\mathbb R}$.
From Lemma 1 in \cite{MZ}, it follows that the Meyer--Zheng topology on the space $\mathbb D[0,T]$
is given by the
metric
$$d_{MZ}(f,g):=\int_{0}^T \min(1,|f(t)-g(t)|)dt+|f(T)-g(T)|, \ \ f,g\in \mathbb D[0,T].$$
We denote the corresponding space by $\mathbb D_{MZ}[0,T]$.

Next, we formulate
our convergence assumptions.
\begin{asm}\label{asm2.3}
For any $k\in\mathbb N$, let $\mathbb D([0,T];\mathbb R^k)$ be the space
of all RCLL functions
$f:[0,T]\rightarrow\mathbb R^k$ equipped with the Skorokhod topology (for details see \cite{B}).
We assume that there exists $m\in\mathbb N$ and a stochastic processes
$X^{n}:\Omega^n\rightarrow \mathbb D([0,T];\mathbb R^m)$, $n\in\mathbb N$,
$X:\Omega\rightarrow C([0,T];\mathbb R^m)$ (i.e. $X$ is continuous) which satisfy the following:
\begin{itemize}
\item[(i)] The filtrations $(\mathcal F^{n}_t)_{0\leq t\leq T}$, $n\in\mathbb N$ and $(\mathcal F_t)_{0\leq t\leq T}$,
are the usual filtrations (right continuous and completed by the corresponding probability measure) generated
by $X^{n}$, $n\in\mathbb N$ and $X$, respectively. \\
\item[(ii)] We have the weak convergence
\begin{equation*}
\left((S^{n}, X^{n}),\mathbb P^n\right)\Rightarrow \left((S,X),\mathbb P\right) \ \ \mbox{on} \ \ \mathbb D([0,T];\mathbb R^{m+1}).
\end{equation*}
The above relation means that the joint distribution of $(S^{n}, X^{n})$ under $\mathbb P^n$ converge to the joint distribution
of $(S,X)$ under $\mathbb P$.\\
\item[(iii)] We have the extended weak convergence
$(X^{n},\mathbb P^n)\Rrightarrow (X,\mathbb P)$. This means (see \cite{A}) that, for any $k$ and a
continuous bounded function $\psi: \mathbb D([0,T];\mathbb R^m)\rightarrow\mathbb R^k$,
we have
\begin{equation*}
\left((X^{n},Y^{n}),\mathbb P^n\right)\Rightarrow \left((X,Y),\mathbb P\right) \ \ \mbox{on} \ \ \mathbb D([0,T];\mathbb R^{m+k}),
\end{equation*}
where
$$Y^{n}_t:=\mathbb E_{\mathbb P^n}\left[\psi(X^{n})|\mathcal F^{n}_t\right] \ \ \mbox{and} \ \
Y_t:=\mathbb E_{\mathbb P}\left[\psi(X)|\mathcal F_t\right], \ \ t\in [0,T].$$
\end{itemize}
\end{asm}

\subsection{The Main Result}
We are ready to state our limit theorem.
\begin{thm}\label{thm1}
Let $x>0$. Then we have,
\begin{equation}\label{values}
u(x)=\lim_{n\rightarrow\infty} u^n(x).
\end{equation}
Moreover, let
 $\hat\gamma^{n}\in\mathcal{A}^{n}(x)$, $n\in\mathbb N$ be a sequence of asymptotically optimal portfolios, namely
\begin{equation}\label{optimal}
\lim_{n\rightarrow\infty} \left(u^n(x)-\mathbb E_{\mathbb P^n}\left[U\left(x+V^{\hat\gamma^{n}}_T\right)\right]\right)=0.
\end{equation}
Then,
\begin{equation}\label{new}
\left((S^{n},\hat\gamma^{n}),\mathbb P^n\right)\Rightarrow \left((S,\gamma^{opt});\mathbb P\right) \ \ \ \mbox{on} \  \mbox{the} \  \mbox{space} \
 \mathbb D([0,T])\times \mathbb D_{MZ}[0,T],
\end{equation}
where $\gamma^{opt}$ is the unique optimal portfolio for the
 optimization problem (\ref{problem}).
\end{thm}
We finish this section with the following remark.
\begin{rem}
Assumptions \ref{asm2.2}, \ref{asm2.3-}, \ref{asm2.3} are analogues (for the current setup) of similar assumptions in \cite{BDD}
and are needed for the proof of (\ref{values}). This proof follows exactly the lines
of the proof from \cite{BDD}. In order to prove the ``new" result (\ref{new}) we
establish a uniqueness result, that is Proposition \ref{lem.unique}. For the proof of Proposition \ref{lem.unique}
we need to assume
Assumptions \ref{asm1}-\ref{asm2}.
\end{rem}
\section{The Uniqueness Result}\label{sec:uni}
In this section we prove that for a given initial capital,
the problem of utility maximization from terminal wealth has a unique optimal trading strategy.
Although, for strictly concave
utility the uniqueness of the optimal terminal wealth is straightforward,
the uniqueness of the optimal trading strategy is far from obvious and was an open question
for the general setup we consider in the present note
(see Remark 6.9 in \cite{S}). It is important to mention the paper
\cite{DKL} where the authors proved a uniqueness result for consumption-investment problems
in the presence of proportional transaction costs where
the price of the assets is given by a geometric L\'{e}vy process.

\begin{prop}\label{lem.unique}
Let $x>0$, be the initial capital. Then, there exists a unique optimal portfolio $\gamma^{opt}=(\gamma^{opt}_t)_{0\leq t\leq T}$ to the
optimization problem (\ref{problem}).
\end{prop}
\begin{proof}
From Theorem 2.3 in \cite{CPSY}, there exists a semi--martingale $\hat S\in [(1-\kappa)S,(1+\kappa)S]$ and
 $\gamma^{1}\in\mathcal A(x)$ such that
$\gamma^{1}$ is a solution to (\ref{problem}) and $(\gamma^{1}_{t-})_{0\leq t\leq T}$ is a solution to the frictionless problem
\begin{equation}\label{1}
\mathbb E_{\mathbb P}\left[U\left(x+\int_{0}^T\gamma^1_{t-} d\hat S_t\right)\right]=\sup_{\theta}\mathbb E_{\mathbb P}\left[U\left(x+\int_{0}^T\theta_{t} d\hat S_t\right)\right]
\end{equation}
where the supremum is taken over all $\hat S$--integrable predictable processes $\theta=(\theta_t)_{0\leq t\leq T}$ which satisfy
the admissibility condition $x+\int_{0}^u \theta_t d\hat{ {S}_t} \geq 0$ for all $u\in [0,T]$.
Moreover, we have
 \begin{equation}\label{2}
\int_{0}^T\gamma^1_{t-} d\hat S_t=V^{\gamma^1}_T.
 \end{equation}
\textbf{Step I:}
Assume by contradiction that there exists an optimal solution $\gamma^{2}\neq\gamma^{1}$ which solves the utility maximization problem
(\ref{problem}).
First, let us notice that $V^{\gamma^1}_T=V^{\gamma^2}_T$. Indeed, if by contraction
there is no equality then from the fact that $U$ is increasing and strictly concave we obtain that for the strategy $\gamma:=(\gamma^1+\gamma^2)/2$
$$\mathbb E_{\mathbb P}\left[U\left(x+V^{\gamma}_T\right)\right]>\frac{\mathbb E_{\mathbb P}\left[U\left(x+V^{\gamma^1}_T\right)\right]+\mathbb E_{\mathbb P}\left[U\left(x+V^{\gamma^2}_T\right)\right]}{2}$$
which contradicts the optimality of $\gamma^{1},\gamma^{2}$.
Thus, from (\ref{2}) and the fact that $\hat S\in [(1-\kappa)S,(1+\kappa)S]$ it follows that
$$\int_{0}^T\gamma^1_{t-} d\hat S_t=V^{\gamma^1}_T=V^{\gamma^2}_T\leq \int_{0}^T\gamma^2_{t-} d\hat S_t.$$
We conclude that
\begin{equation}\label{3}
\int_{0}^T\gamma^1_{t-} d\hat S_t=\int_{0}^T\gamma^2_{t-} d\hat S_t.
\end{equation}
Let us prove that
\begin{equation}\label{4}
\int_{0}^u\gamma^1_{t-} d\hat S_t=\int_{0}^u\gamma^2_{t-} d\hat S_t, \ \ \forall u\in [0,T].
\end{equation}
Assume by contradiction that (\ref{4}) does not hold.
Then, without loss of generality we can assume that there exist $\epsilon>0$, a stopping time
$\Theta\leq T$ and an event of positive probability $A\in\mathcal F_{\Theta}$ such that
\begin{equation}\label{4+}
\int_{0}^{\Theta}\gamma^1_{t-} d\hat S_t-\int_{0}^{\Theta}\gamma^2_{t-} d\hat S_t>\epsilon \ \ \ \ \mbox{on} \ \mbox{the} \ \mbox{event} \ A.
\end{equation}
Define a strategy $(\gamma^3_t)_{0\leq t\leq T}$ by
$$\gamma^3_t:=\gamma^1_t \ \  \mbox{for} \ \  t<\Theta$$
and
$$\gamma^3_t:=(1-\mathbb I_A)\gamma^1_t+\mathbb I_A\gamma^2_t \ \ \mbox{for} \ \ t\geq\Theta.$$
From (\ref{4+}) and the relation
$\hat S\in [(1-\kappa)S,(1+\kappa)S]$
it follows that for any $u\in [0,T]$
\begin{eqnarray*}
&\int_{0}^u\gamma^3_{t-} d\hat S_t\\
&=\left(1-\mathbb I_{A}\mathbb I_{u>\Theta}\right)\int_{0}^u\gamma^1_{t-} d\hat S_t+\mathbb I_{A}\mathbb I_{u>\Theta}
\left(\int_{0}^\Theta\gamma^1_{t-} d\hat S_t+
\int_{\Theta}^u\gamma^2_{t-} d\hat S_t\right)\\
&\geq\left(1-\mathbb I_{A}\mathbb I_{u>\Theta}\right) \int_{0}^u\gamma^1_{t-} d\hat S_t+\mathbb I_{A}\mathbb I_{u>\Theta} \int_{0}^u\gamma^2_{t-} d\hat S_t\\
&\geq\left(1-\mathbb I_{A}\mathbb I_{u>\Theta}\right) V^{\gamma^1}_u+\mathbb I_{A}\mathbb I_{u>\Theta} V^{\gamma^2}_u.
\end{eqnarray*}
Thus, $\gamma^3$ satisfies the
admissibility condition $x+\int_{0}^u \gamma^3_{t-} d\hat{ {S}_t} \geq 0$ for all $u\in [0,T]$. Moreover,
for $u=T$, by applying (\ref{3}) we obtain
\begin{eqnarray*}
&\int_{0}^T\gamma^3_{t-} d\hat S_t\\
&=\left(1-\mathbb I_A\right)\int_{0}^T\gamma^1_{t-} d\hat S_t+\mathbb I_{A}
\left(\int_{0}^\Theta\gamma^1_{t-} d\hat S_t+
\int_{\Theta}^T\gamma^2_{t-} d\hat S_t\right)\\
&\geq \int_{0}^T\gamma^1_{t-} d\hat S_t+\epsilon\mathbb{I}_A.
\end{eqnarray*}
This contradicts the fact that $\gamma^{1}$ is the optimal solution for (\ref{1}) and so, (\ref{4}) follows. \\
${}$\\
\textbf{Step II:}
Since by contradiction $\gamma^1\neq \gamma^2$, there exists $\epsilon>0$ such that the stopping time
$$\sigma=\sigma(\epsilon):=T\wedge\inf\{t:|\gamma^{1}_t-\gamma^{2}_t|>\epsilon\}$$
satisfies
\begin{equation}\label{5}
\mathbb P(\sigma<T)>0.
\end{equation}
Next, define the stopping time
$$\tau:=\inf\{t> \sigma: |\gamma^{1}_t-\gamma^{2}_t|<\epsilon/2\}.$$ Observe that
$\gamma^{1}_T=\gamma^{2}_T=0$ implies $\tau\leq T$ a.s. on the event $\sigma<T$.
Clearly, on the interval $(\sigma,\tau]$ we have $|\gamma^{1}_{t-}-\gamma^{2}_{t-}|\geq\frac{\epsilon}{2}$ and so from
the associativity of the stochastic integral (see Section 2 in \cite{P1})
and (\ref{4}) we conclude that for any $t\in (\sigma,\tau]$
$$\hat S_{\tau}-\hat S_{t}=\int_{t}^{\tau} \frac{1}{\gamma^{1}_{u-}-\gamma^{2}_{u-}} d\left(\int_{0}^u \left(\gamma^1_{v-}-\gamma^2_{v-}\right)d\hat S_v\right)=0.$$
Thus, ($\hat S$ is right continuous)
$\hat S$ is constant on $[\sigma,\tau]$. Since $\hat S\in [(1-\kappa )S,(1+\kappa)S]$ then from Assumption \ref{asm2}
we get $\hat S_{\sigma}\in \left ((1-\kappa)S_{\sigma},(1+\kappa)S_{\sigma})\right)$ (i.e. the shadow price is strictly between the bid price and the ask price).
 From Assumption \ref{asm1} and (\ref{5}), it follows that for the event
 $$B:=\left\{(1-\kappa)S_t<\hat S_{\sigma}<(1+\kappa)S_{t}, \ \ \forall t\in [\sigma,\tau]\right\}$$
 we have $\mathbb P\left(B\cap \{\sigma<T\}\right)>0$.
Finally, since $\hat S$ is constant on the interval $[\sigma,\tau]$,
we observe that on the event $B\cap \{\sigma<T\}$ the interval $[\sigma,\tau]$ is a no-trading region for any solution of (\ref{problem}) (see
Theorem 3.5 in \cite{CS} and Remark 2.13 in \cite{CSY}).
Hence on the event $B\cap \{\sigma<T\}$, $\gamma^{1}_{[\sigma,\tau]}$ and $\gamma^{2}_{[\sigma,\tau]}$ are (random) constants.
In particular
\begin{equation}\label{2.lep}
\gamma^{1}_{\sigma}-\gamma^2_{\sigma}=\gamma^{1}_{\tau}-\gamma^2_{\tau} \ \ \mbox{on} \ \  B\cap \{\sigma<T\}.
\end{equation}
On the other hand, from the definition of $\sigma,\tau$ and the right continuity
of $\gamma^{1},\gamma^{2}$ it follows that
$$|\gamma^{1}_{\sigma}-\gamma^2_{\sigma}|\geq \epsilon \ \ \mbox{and} \ \
|\gamma^{1}_{\tau}-\gamma^2_{\tau}|\leq \epsilon/2 \ \ \mbox{on} \ \ \{\sigma<T\}
$$
which is a contradiction to (\ref{2.lep}).
\end{proof}
\section{Proof of Theorem \ref{thm1} }\label{sec:2}
We start with the following lower semi--continuity result.
\begin{lem}\label{lem2.1}
For any $x>0$ we have $$u(x)\leq\lim\inf_{n\rightarrow \infty} u^n(x).$$
\end{lem}
\begin{proof}
The proof is done by using the same approximating arguments as in Lemma 4.2 in \cite{BDD}.
Observe that since our utility function is not state dependent, then Assumption 2.5(i) in \cite{BDD} is trivially satisfied.
Moreover the continuity of $u$ which is essential for the proof (and was established in Lemma 4.1 in \cite{BDD})
is a well known fact for the current setup (see Theorem  3.2 in \cite{CS}).
\end{proof}
Next, we have the following result.
\begin{lem}\label{lem.tight}
Let $x>0$ and
$\gamma^{(n)}\in\mathcal A^{n}(x)$, $n\in\mathbb N$
be a sequence of admissible trading strategies.
The sequence
$\left((X^{n},S^{n},\gamma^{n}),\mathbb P^n\right)$ is tight on the space
 $\mathbb D([0,T];\mathbb R^{m+1})\times \mathbb D_{MZ}[0,T]$ and so from Prohorov's theorem (see \cite{B})
 it is relatively compact. Moreover, any cluster point is of the form $\left((X,S,\hat\gamma),\mathbb P\right)$
 and satisfies the following conditional independence property:

 Let $(\mathcal F^{X,\hat\gamma}_t)_{0\leq t\leq T}$
 be the usual filtration (right continuous and $\mathbb P$-completed) generated by $X$
 and $\hat\gamma$. Then, for any $t<T$,
 $\mathcal F^{X,\hat\gamma}_t$ and $\mathcal F_T$ are conditionally independent given
 $\mathcal F_t$. As before $\mathcal F:=(\mathcal F_t)_{0\leq t\leq T}$ is the usual filtration
 generated by $X$.
\end{lem}
\begin{proof}
The proof is the same as the proof of Lemma 4.3 in \cite{BDD} and it is based on Assumption \ref{asm2.2}
(Assumption 2.3 in \cite{BDD}) and
the extended weak convergence Assumption \ref{asm2.3} (Assumption 2.9 in \cite{BDD}).
\end{proof}
Now, we are ready to prove Theorem \ref{thm1}.
\begin{proof}
Let $x>0$ and let $\hat\gamma^{n}\in\mathcal{A}^{n}(x)$, $n\in\mathbb N$ be a sequence of portfolios
 which satisfy (\ref{optimal}).
 By passing to a subsequence, we assume without loss of generality that
 $\lim_{n\rightarrow\infty}u^n(x)$
 exists.
${}$\\
\textbf{Step I:}
From Proposition \ref{lem.unique}, there exists a unique solution to (\ref{problem}), denote it by
$\gamma^{opt}$.
From the tightness of the sequence
$\left((S^{n},X^{n},\hat\gamma^{n}),\mathbb P^n\right)$,
 $n\in\mathbb N$ (Lemma \ref{lem.tight}), it follows that
  in order to prove (\ref{new}) it is sufficient to show that the only cluster point of this sequence is
 $(S,X,\gamma^{opt})$.

From Lemma \ref{lem.tight}, any cluster point is of the form $\left((X,S,\hat\gamma),\mathbb P\right)$
where $\hat\gamma$ satisfies the conditional independence
 property which is formulated in this lemma. Let $\hat{\mathcal A}(x)$ be the set of all $(\mathcal F^{X,\hat\gamma}_t)_{0\leq t\leq T}$--adapted
processes
$\gamma=(\gamma_t)_{0\leq t\leq T}$
of
bounded variation with right continuous paths which satisfy
$\gamma_T=0$ and
$x+V^{\gamma}_t\geq 0$, for all $t$. The term $V^{\gamma}$ is defined as in (\ref{2.0}). Introduce the optimization problem
\begin{equation}\label{problem1}
\hat u(x):=\sup_{\gamma\in\hat{\mathcal A}(x)}\mathbb E_{\mathbb P}[U(x+V^{\gamma}_T)].
\end{equation}

By exploiting the uniform integrability result given by Lemma \ref{prop1} (this is Assumption 2.5(ii) in \cite{BDD}),
and applying the same arguments as
in Section 4.2 in \cite{BDD}, we obtain that $\hat\gamma\in\hat{\mathcal A}(x)$
and satisfies
\begin{equation}\label{2.24}
\mathbb E_{\mathbb P}[U(x+V^{\hat\gamma}_T)]\geq \lim_{n\rightarrow\infty} u^n(x).
\end{equation}
Moreover, applying the Jensen inequality and the conditional independence property given by Lemma \ref{lem.tight} (Lemma 4.3 in \cite{BDD})
in the same way as in Section 4.2 in \cite{BDD}, we obtain, for any $\gamma\in\hat{\mathcal A}(x)$, that
\begin{equation}\label{2.53}
V^{\gamma^{\mathcal F}}_t\geq \mathbb E_{\mathbb P}[V^{{\gamma}}_t|\mathcal F_T], \ \ \forall t\in [0,T],
\end{equation}
where $\gamma^{\mathcal F}$
denotes the optional projection of the process $\gamma$ with respect to $\mathcal F$ and it is well
defined. In particular, (\ref{2.53}) implies that $\gamma^{\mathcal F}\in\mathcal A(x)$.

Thus, from the Jensen inequality (for the concave function $U$), (\ref{2.53}) and the trivial relation $\mathcal A(x)\subseteq\hat{\mathcal A}(x)$, we get
\begin{equation}\label{1*}
u(x)\geq \sup_{\gamma\in\hat{\mathcal A}(x)} \mathbb E_{\mathbb P}\left[U\left(x+V^{{\gamma}^{\mathcal F}}_T\right)\right]\geq
 \sup_{\gamma\in\hat{\mathcal A}(x)}\mathcal \mathbb E_{\mathbb P}\left[U\left(x+V^{{\gamma}}_T\right)\right]=\hat u(x)\geq u(x).
\end{equation}
By applying the Jensen inequality,
Lemma \ref{lem2.1} and
(\ref{2.24})--(\ref{2.53}) it follows that
\begin{equation}\label{2*}
u(x)\geq \mathbb E_{\mathbb P}\left[U\left(x+V^{{\hat\gamma}^{\mathcal F}}_T\right)\right]
\geq
\mathbb E_{\mathbb P}\left[U\left(x+V^{\hat\gamma}_T\right)\right]\geq \lim_{n\rightarrow\infty} u^n(x)\geq u(x).
\end{equation}
From (\ref{1*})--(\ref{2*}) we get (\ref{values}) and we conclude that
$\hat\gamma,\gamma^{opt}\in \hat{\mathcal A}(x)$ are optimal portfolios for the
 optimization problem (\ref{problem1}).
 Thus in order complete the proof it remains to argue that the uniqueness result Proposition \ref{lem.unique} holds true where the filtration
$\mathcal F$ is replaced with the filtration $\mathcal F^{X,\hat\gamma}$. For that end it remains to prove that Assumptions \ref{asm1}--\ref{asm2}
hold true with respect to the filtration $\mathcal F^{X,\hat\gamma}$. This brings us to the second step.\\
\textbf{Step II:}
We start with Assumption \ref{asm1}. From Lemma 3.1 in \cite{BPS} it follows that we can restrict $\tau$ in Assumption
\ref{asm1} to be deterministic and the Assumption will remains the same. From \cite{DM} (see Chapter 2, Theorem 45)
and
the conditional independence property given by
Lemma \ref{lem.tight} it follows that for any $t$,
$$\mathbb P(S|\mathcal F_t)=\mathbb P(S|\mathcal F^{X,\hat\gamma}_t)$$
and so
Assumption \ref{asm1} holds true with respect to $\mathcal F^{X,\hat\gamma}$.

Next, we treat Assumption \ref{asm2}.
Assume by contradiction that the Assumption does not hold.
Then, there exists a stopping time with respect to $\mathcal F^{X,\hat\gamma}$,
$\sigma\leq T$ and $\epsilon>0$ such that (without loss of generality we choose the positive direction)
\begin{equation}\label{contrad}
\mathbb P\left(S_t-S_{\sigma}\geq 0 \ \ \forall t\in [\sigma,\sigma+\epsilon]\right)>0.
\end{equation}
By enlarging the underlying probability space we assume (without loss of generality) that there exists a random variable $U$ which is
uniformly distributed on the interval $[0,1]$ and is
independent of
$\mathcal F$. Consider the process
$$Z_t:=\mathbb P(\sigma\leq t|\mathcal F_T), \ \ t\in [0,T].$$
Clearly, $Z$ is a right continuous increasing process which satisfies $Z_T=1$.
Introduce the random time
$$\tau:=\inf\{t:Z_t\geq U\}.$$ Observe that $Z_T=1$ implies that $\tau\leq T$.
Moreover, for any $t\in [0,T]$
$$
\mathbb P(\tau\leq t|\mathcal F_T)=\mathbb P(Z_t\geq U|\mathcal F_T)=Z_t=\mathbb P(\sigma\leq t|\mathcal F_T).
$$
We conclude,
\begin{equation}\label{dist}
((S,\sigma);\mathbb P)=((S,\tau);\mathbb P).
\end{equation}

Next, for any $u\in [0,1]$ define the random time
$\tau_u:=\inf\{t:Z_t\geq u\}$.
From the conditional independence property given by Lemma \ref{lem.tight}, it follows that
$$Z_t=\mathbb P(\sigma\leq t|\mathcal F_t), \ \ \forall t\in [0,T].$$
Hence,
for any $u\in [0,1]$, $\tau_u$
is a stopping time with respect to the filtration $\mathcal F$.
From (\ref{dist}) and the fact that $U$ is independent of $S$ it follows that
\begin{eqnarray*}
&\mathbb P\left(S_t-S_{\sigma}\geq 0 \ \ \forall t\in [\sigma,\sigma+\epsilon]\right)\\
&\mathbb P\left(S_t-S_{\tau}\geq 0 \ \ \forall t\in [\tau,\tau+\epsilon]\right)\\
&=\int_{0}^1 \mathbb P\left(S_t-S_{\tau_u}\geq 0 \ \ \forall t\in [\tau_u,\tau_u+\epsilon]\right)du=0
\end{eqnarray*}
where the last equality follows from Assumption \ref{asm2} (for the filtration $\mathcal F$).
We obtain a contradiction to (\ref{contrad}), which
completes the proof.
\end{proof}

\bibliographystyle{spbasic}

\end{document}